\documentclass[11pt]{article}
\usepackage[a4paper]{geometry}
\usepackage{amsfonts, amsmath, amssymb, amsthm, graphicx, caption, authblk, multirow, makecell, framed, float, xcolor, enumitem}
\theoremstyle{definition}
\newtheorem{theorem}{Theorem}
\newtheorem{definition}{Definition}
\newtheorem{example}{Example}

\newtheorem{lemma}{Lemma}

\theoremstyle{remark}
\newtheorem*{remark}{Remark}

\begin{document}
\title{Random Popular Matchings with Incomplete Preference Lists\thanks{This paper is an extended version of \cite{ruangwises}, which appeared at WALCOM 2018.}}
\author[1]{Suthee Ruangwises\thanks{\texttt{ruangwises.s.aa@m.titech.ac.jp}}}
\author[1]{Toshiya Itoh\thanks{\texttt{titoh@c.titech.ac.jp}}}
\affil[1]{Department of Mathematical and Computing Science, Tokyo Institute of Technology, Tokyo, Japan}
\date{}
\maketitle

\begin{abstract}
Given a set $A$ of $n$ people and a set $B$ of $m \geq n$ items, with each person having a list that ranks his/her preferred items in order of preference, we want to match every person with a unique item. A matching $M$ is called \textit{popular} if for any other matching $M'$, the number of people who prefer $M$ to $M'$ is not less than the number of those who prefer $M'$ to $M$. For given $n$ and $m$, consider the probability of existence of a popular matching when each person's preference list is independently and uniformly generated at random. Previously, Mahdian \cite{mahdian} showed that when people's preference lists are \textit{strict} (containing no ties) and \textit{complete} (containing all items in $B$), if $\alpha = m/n > \alpha_*$, where $\alpha_* \approx 1.42$ is the root of equation $x^2 = e^{1/x}$, then a popular matching exists with probability $1-o(1)$; and if $\alpha < \alpha_*$, then a popular matching exists with probability $o(1)$, i.e. a phase transition occurs at $\alpha_*$. In this paper, we investigate phase transitions in the case that people's preference lists are strict but not complete. We show that in the case where every person has a preference list with length of a constant $k \geq 4$, a similar phase transition occurs at $\alpha_k$, where $\alpha_k \geq 1$ is the root of equation $x e^{-1/2x} = 1-(1-e^{-1/x})^{k-1}$.

\end{abstract}

\section{Introduction}

A simple problem of matching people with items, with each person having a list that ranks his/her preferred items, models many important real-world situations such as the assignment of graduates to training positions \cite{hylland}, families to government-subsidized housing \cite{yuan}, and DVDs to subscribers \cite{mahdian}. The main target of such problems is to find an ``optimal'' matching in each situation. Various definitions of optimality have been proposed. The least restrictive one is \textit{Pareto optimality} \cite{abdulkadiroglu,abraham1,roth}. A matching $M$ is Pareto optimal if there is no other matching $M'$ such that at least one person prefers $M'$ to $M$ but no one prefers $M$ to $M'$. Other stronger definitions include \textit{rank-maximality} \cite{irving} (allocating maximum number of people to their first choices, then maximum number to their second choices, and so on), and \textit{popularity} \cite{abraham,gardenfors} defined below.

\subsection{Popular Matching}

Consider a set $A$ of $n$ people and a set $B$ of $m \geq n$ items, with $\alpha = m/n \geq 1$. Each person in $A$ has a preference list that ranks some items in $B$ in order of preference. A preference list is \textit{strict} if it does not contain ties, and is \textit{complete} if it contains all items in $B$. Each person can only be matched with an item in his/her preference list, and each item can be matched with at most one person.

For a matching $M$, a person $a \in A$, and an item $b \in B$, let $M(a)$ denote an item matched with $a$, and $M(b)$ denote a person matched with $b$ (for convenience, let $M(a)$ be \textit{null} for an unmatched person $a$). Let $r_a(b)$ denote the rank of item $b$ in $a$'s preference list, with the most preferred item having rank 1, the second most preferred item having rank 2, and so on (for convenience, let $r_a(null) = \infty$). For any pair of matchings $M$ and $M'$, define $\phi(M,M')$ to be the number of people who prefer $M$ to $M'$, i.e. $\phi(M,M') = |\{a \in A|r_a(M(a)) < r_a(M'(a))\}|$. We say that a matching $M$ is \textit{popular} if $\phi(M,M') \geq \phi(M',M)$ for every other matching $M'$. As the relation $\phi(M,M') \geq \phi(M',M)$ is not transitive, a popular matching may or may not exist depending on the preference lists of people. See Example \ref{ex}.

A probabilistic variant of this problem, the \textit{Random Popular Matching Problem} (\textsc{rpmp}), studies the probability that a popular matching exists in a random instance when given $n$ and $m$, and each person's preference list is defined independently by selecting the first item $b_1 \in B$ uniformly at random, the second item $b_2 \in B\setminus \{b_1\}$ uniformly at random, the third item $b_{3} \in B \setminus\{b_{1},b_{2}\}$ uniformly at random, and so on.

\begin{example} \label{ex}
Consider the following instance with three people $a_1,a_2,a_3$ and three items $b_1,b_2,b_3$, with everyone having the same preferences.

\begin{center}
	\begin{minipage}{0.3\textwidth}
        \underline{Preference Lists} \\
        $\boldsymbol{a_1:} \hspace{0.2cm} b_1, b_2, b_3$ \\
        $\boldsymbol{a_2:} \hspace{0.2cm} b_1, b_2, b_3$ \\
        $\boldsymbol{a_3:} \hspace{0.2cm} b_1, b_2, b_3$ \\
    \end{minipage}
    \begin{minipage}{0.4\textwidth}
        $M_1 = \{\{a_1,b_1\}, \{a_2,b_2\}, \{a_3,b_3\}\}$ \\
				$M_2 = \{\{a_1,b_2\}, \{a_2,b_3\}, \{a_3,b_1\}\}$ \\
        $M_3 = \{\{a_1,b_3\}, \{a_2,b_1\}, \{a_3,b_2\}\}$ \\
    \end{minipage}
\end{center}

For the three above matchings, we have $\phi(M_1,M_2) = 2 > 1 = \phi(M_2,M_1)$. Similarly, we also have $\phi(M_2,M_3) = 2 > 1 = \phi(M_3,M_2)$ and $\phi(M_3,M_1) = 2 > 1 = \phi(M_1,M_3)$. In fact, a popular matching does not exist in this instance.
\end{example}

\subsection{Related Work}

The concept of popularity of a matching was first introduced by Gardenfors \cite{gardenfors} in the context of the \textit{Stable Marriage Problem}. Abraham et al. \cite{abraham} presented the first polynomial time algorithm to find a popular matching in a given instance, or to report that none exists. Later, Mestre \cite{mestre} generalized that algorithm to the case where people are given different voting weights. Manlove and Sng \cite{manlove} presented an algorithm to determine whether a popular matching exists in a setting known as the \textit{Capacitated House Allocation Problem}, which allows an item to be matched with more than one person. The notion of popularity also applies when the preference lists are two-sided (matching people with people), both in a bipartite graph (\textit{Marriage Problem}) and general graph (\textit{Roommates Problem}). Bir\'{o} et al. \cite{biro} developed an algorithm to test popularity of a matching in these two settings, and proved that determining whether a popular matching exists in these settings is NP-hard when ties are allowed.

While a popular matching does not always exist, McCutchen \cite{mccutchen} introduced two measures of ``unpopularity" of a matching, the \textit{unpopularity factor} and the \textit{unpopularity margin}, and showed that finding a matching that minimizes either measure is NP-hard. Huang et al. \cite{huang} later gave algorithms to find a matching with bounded values of these measures in certain instances. Kavitha et al. \cite{kavitha} introduced the concept of a \textit{mixed matching}, which is a probability distribution over matchings, and proved that a mixed matching that is popular always exists.

For \textsc{rpmp} in the case with strict and complete preference lists, Mahdian \cite{mahdian} proved that if $\alpha = m/n > \alpha_*$, where $\alpha_* \approx 1.42$ is the root of equation $x^2 = e^{1/x}$, then a popular matching exists with high ($1-o(1)$) probability in a random instance. On the other hand, if $\alpha < \alpha_*$, a popular matching exists with low ($o(1)$) probability. The point $\alpha = \alpha_*$ can be regarded as a phase transition point, at which the probability rises from asymptotically zero to asymptotically one. Itoh and Watanabe \cite{itoh} later studied the weighted case where each person has weight either $w_1$ or $w_2$, with $w_1 \geq 2w_2$, and found a phase transition at $\alpha = \Theta(n^{1/3})$.

\subsection{Our Contribution}

\textsc{rpmp} in the case that preference lists are strict but not complete, with every person's preference list having the same length of a constant $k$ was simulated by Abraham et al. \cite{abraham}, and was conjectured by Mahdian \cite{mahdian} that the phase transition will shift by an amount exponentially small in $k$. However, the exact phase transition point, or whether it exists at all, had not been found yet. In this paper, we study this case and prove a phase transition at $\alpha = \alpha_k$, where $\alpha_k \geq 1$ is the root of equation $x e^{-1/2x} = 1-(1-e^{-1/x})^{k-1}$. In particular, we prove that for $k \geq 4$, if $\alpha > \alpha_k$, then a popular matching exists with high probability; and if $\alpha < \alpha_k$, then a popular matching exists with low probability. For $k \leq 3$, where the equation does not have a solution in $[1, \infty)$, a popular matching always exists with high probability for any value of $\alpha \geq 1$ without a phase transition.

\section{Preliminaries}

For convenience, for each person $a \in A$ we append a unique auxiliary \textit{last resort item} $\ell_a$ to the end of $a$'s preference list ($\ell_a$ has lower preference than all other items in the list). By introducing the last resort items, we can assume that every person is matched because we can simply match any unmatched person $a$ with $\ell_a$. Note that these last resort items are not in $B$ and do not count toward $m$, the total number of ``real items.'' Also, let $L = \{\ell_a|a \in A\}$ be the set of all last resort items.

For each person $a \in A$, let $f(a)$ denote the item at the top of $a$'s preference list. Let $F$ be the set of all items $b \in B$ such that there exists a person $a' \in A$ with $f(a')=b$, and let $S = B-F$. Then, for each person $a \in A$, let $s(a)$ denote the highest ranked item in $a$'s preference list that is not in $F$. Note that $s(a)$ is well-defined for every $a \in A$ because of the existence of last resort items.

We say that a matching $M$ is \textit{$A$-perfect} if every person $a \in A$ is matched with either $f(a)$ or $s(a)$. Abraham et al. \cite{abraham} proved the following lemma, which holds for any instance with strict (not necessarily complete) preference lists.

\begin{lemma} \label{pmexists}
\cite{abraham} In an instance with strict preference lists, a popular matching exists if and only if an $A$-perfect matching exists.
\end{lemma}

The proof of Lemma \ref{pmexists} first shows that a matching $M$ is popular if and only if $M$ is an $A$-perfect matching such that every item in $F$ is matched in $M$. This equivalence implies the forward direction of the lemma. On the other hand, if an $A$-perfect matching $M$ exists in an instance, the proof shows that we can modify $M$ to make every item in $F$ matched, hence implying the backward direction of the lemma.

It is worth noting another useful lemma about independent and uniform selection of items at random proved by Mahdian \cite{mahdian}, which will be used throughout this paper.

\begin{lemma} \label{pick}
\cite{mahdian} Suppose that we pick $y$ elements from the set $\{1,...,z\}$ independently and uniformly at random (with replacement). Let a random variable $X$ be the number of elements in the set that are not picked. Then, $\mathbb{E}[X]=e^{-y/z}z-\Theta(1)$ and $\text{Var}[X] < \mathbb{E}[X]$.
\end{lemma}

\section{Complete Preference Lists Setting}

We first consider the setting that every person's preference list is strict and complete. Note that when $m > n$ and the preference lists are complete, the last resort items are not necessary.

From a given instance, we construct a \textit{top-choice graph}, a bipartite graph with parts $B' = B$ and $S' = S$ such that each person $a \in A$ corresponds to an edge connecting $f(a) \in B'$ and $s(a) \in S'$. Note that multiple edges are allowed in this graph. Previously, Mahdian \cite{mahdian} proved the following lemma.

\begin{lemma} \label{complexcomp1}
\cite{mahdian} In an instance with strict and complete preference lists, an $A$-perfect matching exists if and only if its top-choice graph does not contain a \textit{complex component}, i.e. a connected component with more than one cycle.
\end{lemma}

By Lemmas \ref{pmexists} and \ref{complexcomp1}, the problem of determining whether a popular matching exists is equivalent to determining whether the top-choice graph contains a complex component. However, the difficulty is that the number of vertices in the randomly generated top-choice graph is not fixed. Therefore, a random bipartite graph $G(x,y,z)$ with fixed number of vertices is defined as follows to approximate the top-choice graph.

\begin{definition} \label{defG1}
For integers $x,y,z$, $G(x,y,z)$ is a bipartite graph with $V \cup U$ as a set of vertices, where $V = \{v_1,v_2,...,v_x\}$ and $U = \{u_1,u_2,...,u_y\}$. Each of the $z$ edges of $G(x,y,z)$ is selected independently and uniformly at random (with replacement) from the set of all possible edges between a vertex in $V$ and a vertex in $U$.
\end{definition}

This auxiliary graph has properties closely related to the top-choice graph. Mahdian \cite{mahdian} then proved that if $\alpha > \alpha_* \approx 1.42$, then $G(m,h,n)$ contains a complex component with low probability for any integer $h \in [e^{-1/\alpha}m-m^{2/3}, e^{-1/\alpha}m+m^{2/3}]$, and used those properties to conclude that the top-choice graph also contains a complex component with low probability, hence a popular matching exists with high probability.

\begin{theorem} \label{thmmahdian1}
\cite{mahdian} In a random instance with strict and complete preference lists, if $\alpha > \alpha_*$, where $\alpha_* \approx 1.42$ is the solution of the equation $x^2e^{-1/x} = 1$, then a popular matching exists with probability $1-o(1)$.
\end{theorem}

Theorem \ref{thmmahdian1} serves as an upper bound of the phase transition point in the case of strict and complete preference lists. On the other hand, the following lower bound was also proposed by Mahdian \cite{mahdian} along with a sketch of the proof, although the fully detailed proof was not given.

\begin{theorem} \label{thmmahdian2}
\cite{mahdian} In a random instance with strict and complete preference lists, if $\alpha < \alpha_*$, then a popular matching exists with probability $o(1)$.
\end{theorem}

\section{Incomplete Preference Lists Setting}

The previous section shows known results in the setting that preference lists are strict and complete. However, preference lists in many real-world situations are not complete, as people may regard only some items as acceptable for them. In the setting that the preference lists are strict but not complete, we will consider the case that every person's preference list has equal length of a constant $k \leq m$ (not counting the last resort item). Such instance is called an \textit{instance with $k$-incomplete preference lists}.

\begin{definition} \label{defrandincomp}
For a positive integer $k \leq m$, a \textit{random instance with strict and $k$-incomplete preference lists} is an instance with each person's preference list chosen independently and uniformly from the set of all $\frac{m!}{(m-k)!}$ possible $k$-permutations of the $m$ items in $B$ at random.
\end{definition}

Recall that $F = \{ b \in B | \exists a' \in A, f(a') = b \}$, $S = B-F$, and for each person $a \in A$, $s(a)$ is the highest ranked item in $a$'s preference list not in $F$. The main difference from the complete preference lists setting is that in the incomplete preference lists setting, $s(a)$ can be either a real item or the last resort item $\ell_a$. For each person $a \in A$, let $P_a$ be the set of items in $a$'s preference list (not including the last resort item $\ell_a$). We then define $A_1 = \{ a \in A | P_a \subseteq F \}$ and $A_2 = \{ a \in A | P_a \nsubseteq F \}$. Note that $s(a) = \ell_a$ if and only if $a \in A_1$.

\subsection{Top-Choice Graph}

Analogously to the complete preference lists setting, we define the top-choice graph of an instance with strict and $k$-incomplete preference lists to be a bipartite graph with parts $B' = B$ and $S' \cup L'$, where $S' = S$ and $L'$ = $L$. Each person $a \in A_2$ corresponds to an edge connecting $f(a) \in B'$ and $s(a) \in S'$. We call these edges \textit{normal edges}. Each person $a \in A_1$ corresponds to an edge connecting $f(a) \in B'$ and $s(a) = \ell_a \in L'$. We call these edges \textit{last resort edges}.

Although the statement of Lemma \ref{complexcomp1} proved by Mahdian \cite{mahdian} is for the complete preference lists setting, exactly the same proof applies to incomplete preference lists setting as well. The proof first shows that an $A$-perfect matching exists if and only if each edge in the top-choice graph can be oriented such that each vertex has at most one incoming edge (because if an $A$-perfect matching $M$ exists, we can orient each edge corresponding to $a \in A$ toward the endpoint corresponding to $M(a)$, and vice versa). Then, the proof shows that for any undirected graph $H$, each edge of $H$ can be oriented in such a manner if and only if $H$ does not have a complex component. Thus we can conclude the following lemma.

\begin{lemma} \label{complexcomp2}
In an instance with strict and $k$-incomplete preference lists, an $A$-perfect matching exists if and only if its top-choice graph does not contain a complex component.
\end{lemma}

In contrast to the complete preference lists setting, the top-choice graph in the incomplete preference lists setting has two types of edges (normal edges and last resort edges) with different distributions, and thus cannot be approximated by $G(x,y,z)$ defined in the previous section. Therefore, we have to construct another auxiliary graph $G'(x,y,z_1,z_2)$ as follows.

\begin{definition} \label{defG2}
For integers $x,y,z_1,z_2$, $G'(x,y,z_1,z_2)$ is a bipartite graph with $V \cup U \cup U'$ as a set of vertices, where $V = \{v_1,v_2,...,v_x\}$, $U = \{u_1,u_2,...,u_y\}$, and $U' = \{u'_1, u'_2, ..., u'_{z_1+z_2}\}$. This graph has $z_1+z_2$ edges. Each of the first $z_1$ edges is selected independently and uniformly at random (with replacement) from the set of all possible edges between a vertex in $V$ and a vertex in $U$. Then, each of the next $z_2$ edges is constructed by the following procedures: Uniformly select a vertex $v_i$ from $V$ at random (with replacement); then, uniformly select a vertex $u'_j$ that has not been selected before from $U'$ at random (without replacement) and construct an edge ($v_i, u'_j)$.
\end{definition}

The intuition of $G'(x,y,z_1,z_2)$ is that we imitate the distribution of the top-choice graph in the incomplete preference list setting, with $V$, $U$, and $U'$ correspond to $B'$, $S'$, and $L'$, respectively, and the first $z_1$ edges and the next $z_2$ edges correspond to normal edges and last resort edges, respectively.

Similarly to the complete preference lists setting, this auxiliary graph has properties closely related to the top-choice graph in incomplete preference lists setting, as shown in the following lemma.

\begin{lemma} \label{lemG2}
Suppose that $\alpha = m/n$, the top-choice graph $H$ has $t$ normal edges and $n-t$ last resort edges for a fixed integer $t \leq n$, and $E$ is an arbitrary event defined on graphs. If the probability of $E$ on the random graph $G'(m,h,t,n-t)$ is at most $O(1/n)$ for every fixed integer $h \in [e^{-1/\alpha}m-m^{2/3},e^{-1/\alpha}m+m^{2/3}]$, then the probability of $E$ on the top-choice graph $H$ is at most $O(n^{-1/3})$.
\end{lemma}

\begin{proof}
Using the same technique as in Mahdian's proof of \cite[Lemma 3]{mahdian}, let a random variable $X$ be the number of isolated vertices (zero-degree vertices) in part $V$ (the part that has $m$ vertices) of $G'(m,h,t,n-t)$. By the definition of $G'(m,h,t,n-t)$, for each fixed value of $h$, the distribution of $H$ conditioned on $|S'| = h$ is the same as the distribution of $G'(m,h,t,n-t)$ conditioned on $X = h$ (because $|S| = |S'| = h$ means that part $B'$ of $H$ has exactly $h$ isolated vertices which correspond to the vertices in $S$). Also, from Lemma \ref{pick} with $y = n$ and $z = m$, we have $\mathbb{E}[X] = e^{-1/\alpha}m - \Theta(1)$ and $\text{Var}[X] < \mathbb{E}[X]$. Let $\delta = \frac{1}{2}m^{2/3}$, and let $I = [E[X]-\delta, E[X]+\delta]$. We have $I \subseteq [e^{-1/\alpha}m-m^{2/3},e^{-1/\alpha}m+m^{2/3}]$ for large enough $m$. Therefore,
\begin{align*}
\Pr_{H}[E] &= \sum_h \Pr_{H}\left[E \big| |S|=h\right] \cdot \Pr_{H}[|S|=h] \\
&= \sum_h \Pr_{G'(m,h,t,n-t)}[E|X=h] \cdot \Pr_{G'(m,h,t,n-t)}[X=h] \\
&= \sum_h \Pr_{G'(m,h,t,n-t)}[X=h|E] \cdot \Pr_{G'(m,h,t,n-t)}[E] \\
&\leq \Pr[|X-\mathbb{E}[X]| > \delta] + \sum_{h \in I} \Pr_{G'(m,h,t,n-t)}[X=h|E] \cdot \Pr_{G'(m,h,t,n-t)}[E] \\
&\leq \Pr[|X-\mathbb{E}[X]| > \delta] + \sum_{h \in I} \Pr_{G'(m,h,t,n-t)}[E].
\end{align*}

From Chebyshev's inequality, we have
\begin{align*}
\Pr_{H}[E] &\leq \frac{\text{Var}[X]}{\delta^2} + \sum_{h \in I} \Pr_{G'(m,h,t,n-t)}[E] \\
&\leq \frac{\mathbb{E}[X]}{\delta^2} + 2\delta \max_{h \in I} \Pr_{G'(m,h,t,n-t)}[E] \\
&< \frac{O(m)}{m^{4/3}} + m^{2/3}O\left(\frac{1}{n}\right) \\
&= O(n^{-1/3})
\end{align*}
as desired.
\end{proof}

\subsection{Size of \boldmath{$A_2$}}

Since our top-choice graph has two types of edges with different distributions, we first want to bound the number of each type of edges. Note that the top-choice graph has $|A_2|$ normal edges and $|A_1|$ last resort edges, so the problem is equivalent to bounding the size of $A_2$.

First, we will prove the next two lemmas, which will be used to bound the ratio $\frac{|A_2|}{n}$.

\begin{lemma} \label{beta1}
In a random instance with strict and $k$-incomplete preference lists,
$$1-e^{-1/\alpha}-c_1 < \frac{|F|}{m} < 1-e^{-1/\alpha}+c_1$$
with probability $1-o(1)$ for any constant $c_1 > 0$.
\end{lemma}

\begin{proof}
Let $c_1 > 0$ be any constant. From Lemma \ref{pick} with $y = n$ and $z = m$, we have 
\begin{align}
\mathbb{E}[|F|] &= m-\mathbb{E}[|S|] = (1-e^{-1/\alpha})m + \Theta(1); \label{ineq3} \\
\text{Var}(|F|) &= \text{Var}(|S|) < \mathbb{E}[|S|] \leq \frac{e^{-1/\alpha}}{1-e^{-1/\alpha}}\mathbb{E}[|F|]. \nonumber
\end{align}
From Chebyshev's inequality, we have
\begin{align}
\Pr\left[\big||F|-\mathbb{E}[|F|]\big| \geq c_1 \cdot \mathbb{E}[|F|]\right] &\leq \frac{\text{Var}[|F|]}{\left(c_1 \cdot \mathbb{E}[|F|]\right)^2} \nonumber \\
&< \frac{e^{-1/\alpha}}{c_1^2(1-e^{-1/\alpha}) \mathbb{E}[|F|]} = O(1/n). \label{ineq3.5}
\end{align}
Therefore, from (\ref{ineq3}) and (\ref{ineq3.5}) we can conclude that
$$1-e^{-1/\alpha}-c_1 < \frac{|F|}{m} < 1-e^{-1/\alpha}+c_1$$
with probability $1-o(1)$ for sufficiently large $m$.
\end{proof}

\begin{lemma} \label{beta2}
In a random instance with strict and $k$-incomplete preference lists,
$$1-(1-e^{-1/\alpha})^{k-1}-c_2 < \Pr[a \in A_2] < 1-(1-e^{-1/\alpha})^{k-1}+c_2$$
holds for any $a \in A$ for sufficiently large $m$, given any constant $c_2 > 0$.
\end{lemma}

\begin{proof}
If $k = 1$, then we have $P_a \subseteq F$ for every $a \in A$, which means $\Pr[a \in A_2] = 0$ and thus the lemma holds. From now on, we will consider the case that $k \geq 2$.

Let $c_2 > 0$ be any constant. We can select a sufficiently small $c_1$ (e.g. $c_1 = \frac{c_2}{(k-1)(c_2+2)}$, where the proof is given in Appendix \ref{Appendix-1}) such that
\begin{align} 
(1-e^{-1/\alpha} - c_1)^{k-1} &> (1-e^{-1/\alpha})^{k-1} - \frac{c_2}{2}; \label{ineq1} \\
(1-e^{-1/\alpha} + c_1)^{k-1} &< (1-e^{-1/\alpha})^{k-1} + \frac{c_2}{2}, \label{ineq2}.
\end{align}

Let $I = [(1-e^{-1/\alpha}-c_1)m, (1-e^{-1/\alpha}+c_1)m]$. From Lemma \ref{beta1}, $|F| \in I$ with probability $1-o(1)$ for sufficiently large $m$.

Note that $a \in A_1$ if and only if $P_a-\{f(a)\} \subseteq F$. Consider the process that we first independently and uniformly select the first-choice item of every person in $A$ from the set $B$ at random, creating the set $F$. Suppose that $|F| = q$ for some fixed integer $q \in I$. Then, for each $a \in A$, we uniformly select the remaining $k-1$ items in $a$'s preference list one by one from the remaining $m-1$ items in $B-\{f(a)\}$ at random. Among the $(k-1)! \binom{m-1}{k-1}$ possible ways of selection, there are $(k-1)! \binom{q-1}{k-1}$ ways such that $P_a-\{f(a)\} \subseteq F$, so
\begin{align*}
\Pr\left[a \in A_1 \big| |F| = q\right] &= \Pr\left[P_a-\{f(a)\} \subseteq F \big| |F| = q\right] \\
&= \frac{(k-1)! \binom{q-1}{k-1}}{(k-1)! \binom{m-1}{k-1}} \\
&= \frac{\binom{q-1}{k-1}}{\binom{m-1}{k-1}}.
\end{align*}

Since $\binom{q-1}{k-1}/\binom{m-1}{k-1}$ converges to $\left(\frac{q}{m}\right)^{k-1}$ when $m$ increases to infinity for every $q \in I$, it is sufficient to consider $\Pr\left[a \in A_1 \big| |F| = q\right] = \left(\frac{q}{m}\right)^{k-1}$. 

Now consider
\begin{align*}
\Pr[a \in A_1] &= \sum_q \Pr[|F| = q] \cdot \Pr\left[a \in A_1 \big| |F| = q\right] \\
&= \sum_{q \in I} \Pr[|F| = q] \cdot \Pr\left[a \in A_1 \big| |F| = q\right] \\
&~~~~~~~~~~+ \sum_{q \notin I} \Pr[|F| = q] \cdot \Pr\left[a \in A_1 \big| |F| = q\right].
\end{align*}

For the lower bound of $\Pr[a \in A_1]$, we have
\begin{align*}
\Pr[a \in A_1] &\geq \sum_{q \in I} \Pr[|F| = q] \cdot \Pr\left[a \in A_1 \big| |F| = q\right] \\
&= \sum_{q \in I} \Pr[|F| = q] \cdot \left(\frac{q}{m}\right)^{k-1} \\
&\geq \sum_{q \in I} \Pr[|F| = q] \cdot (1-e^{-1/\alpha}-c_1)^{k-1} \\
&= \Pr[|F| \in I] \cdot (1-e^{-1/\alpha}-c_1)^{k-1} \\
&> (1-o(1))\left((1-e^{-1/\alpha})^{k-1} - \frac{c_2}{2}\right),
\end{align*}
where the last inequality follows from (\ref{ineq1}). Thus, we can conclude that $\Pr[a \in A_1] > (1-e^{-1/\alpha})^{k-1} - c_2$ for sufficiently large $m$. On the other hand, for the upper bound of $\Pr[a \in A_1]$, we have
\begin{align*}
\Pr[a \in A_1] &\leq \sum_{q \in I} \Pr[|F| = q] \cdot \Pr\left[a \in A_1 \big| |F| = q\right] 
+ \sum_{q \notin I} \Pr[|F| = q] \\
&= \sum_{q \in I} \Pr[|F| = q] \cdot \left(\frac{q}{m}\right)^{k-1} + o(1) \\
&\leq \sum_{q \in I} \Pr[|F| = q] \cdot (1-e^{-1/\alpha}+c_1)^{k-1} + o(1) \\
&= \Pr[|F| \in I] \cdot (1-e^{-1/\alpha}+c_1)^{k-1} + o(1) \\
&< (1-o(1))\left((1-e^{-1/\alpha})^{k-1} + \frac{c_2}{2}\right) + o(1),
\end{align*}
where the last inequality follows from (\ref{ineq2}). Thus, we can conclude that $\Pr[a \in A_1] < (1-e^{-1/\alpha})^{k-1} + c_2$ for sufficiently large $m$.

Therefore,
$$(1-e^{-1/\alpha})^{k-1}-c_2 < \Pr[a \in A_1] < (1-e^{-1/\alpha})^{k-1}+c_2,$$
which is equivalent to
$$1-(1-e^{-1/\alpha})^{k-1}-c_2 < \Pr[a \in A_2] < 1-(1-e^{-1/\alpha})^{k-1}+c_2.$$
\end{proof}

Finally, the following lemma shows that the ratio $\frac{|A_2|}{n}$ lies around a constant $1-(1-e^{-1/\alpha})^{k-1}$ with high probability.

\begin{lemma} \label{beta}
In a random instance with strict and $k$-incomplete preference lists,
$$1-(1-e^{-1/\alpha})^{k-1}-c_3 < \frac{|A_2|}{n} < 1-(1-e^{-1/\alpha})^{k-1}+c_3$$
with probability $1-o(1)$ for any constant $c_3 > 0$.
\end{lemma}

\begin{proof}
If $k = 1$, then we have $P_a \subseteq F$ for every $a \in A$, which means $|A_2| = 0$ and thus the lemma holds. From now on, we will consider the case that $k \geq 2$.

Let $c_3 > 0$ be any constant. We can select a sufficiently small $c_2$ such that $c_2(1+(1-e^{-1/\alpha})^{k-1}+c_2) < c_3$ and thus
\begin{align} 
(1-c_2)\left((1-e^{-1/\alpha})^{k-1}-c_2\right) &> (1-e^{-1/\alpha})^{k-1} - c_3; \label{ineq1a} \\
(1+c_2)\left((1-e^{-1/\alpha})^{k-1}+c_2\right) &< (1-e^{-1/\alpha})^{k-1} + c_3; \label{ineq2a}
\end{align}
From Lemma \ref{beta2}, we have
\begin{equation} \label{ineq4}
1-(1-e^{-1/\alpha})^{k-1}-c_2 < \Pr[a \in A_2] < 1-(1-e^{-1/\alpha})^{k-1}+c_2
\end{equation}
for sufficiently large $m$.

For each $a \in A$, define an indicator random variable $X_a$ such that
$$
X_a = \begin{cases}
1, &\text{for } a \in A_2; \\
0, &\text{for } a \notin A_2.
\end{cases}
$$
Note that $|A_2| = \sum_{a \in A}X_a$. From (\ref{ineq4}), we have
$$1-(1-e^{-1/\alpha})^{k-1}-c_2 < \mathbb{E}[X_a] < 1-(1-e^{-1/\alpha})^{k-1}+c_2$$
for each $a \in A$, and from the linearity of expectation we also have
\begin{equation} \label{ineq5}
\left(1-(1-e^{-1/\alpha})^{k-1}-c_2\right)n < \mathbb{E}[|A_2|] < \left(1-(1-e^{-1/\alpha})^{k-1}+c_2\right)n.
\end{equation}

Since $X_a$ and $X_{a'}$ are independent for any pair of distinct $a, a' \in A$, we have
\begin{align*}
\text{Var}[|A_2|] &= \sum_{a \in A} \text{Var}[X_a] = \sum_{a \in A} \left(\mathbb{E}[X_a^2]-\mathbb{E}[X_a]^2\right) \\
&\leq \sum_{a \in A} \mathbb{E}[X_a^2] = \sum_{a \in A} \mathbb{E}[X_a] = \mathbb{E}[A_2].
\end{align*}
Then, from Chebyshev's inequality and (\ref{ineq5}) we have
$$\Pr\left[\big||A_2|-\mathbb{E}[|A_2|]\big| \geq c_2 \cdot \mathbb{E}[|A_2|]\right]
\leq \frac{\text{Var}[|A_2|]}{\left(c_2 \cdot \mathbb{E}[|A_2|]\right)^2}
\leq \frac{1}{c_2^2 \cdot \mathbb{E}[|A_2|]}
= O(1/n).$$
This implies $(1-c_2)\mathbb{E}[|A_2|] \leq |A_2| \leq (1+c_2)\mathbb{E}[|A_2|]$ with probability $1-O(1/n) = 1-o(1)$. Therefore, from (\ref{ineq1a}), (\ref{ineq2a}), and (\ref{ineq5}) we can conclude that
$$1-(1-e^{-1/\alpha})^{k-1}-c_3 < \frac{|A_2|}{n} < 1-(1-e^{-1/\alpha})^{k-1}+c_3$$
with probability $1-o(1)$
\end{proof}

\section{Main Results}

For each value of $k$, we want to find a phase transition point $\alpha_k$ such that if $\alpha > \alpha_k$, then a popular matching exists with high probability; and if $\alpha < \alpha_k$, then a popular matching exists with low probability. We do so by proving the upper bound and lower bound separately.

\subsection{Upper Bound}

\begin{lemma} \label{lemupperbound}
Suppose that $\alpha = m/n$ and $0 \leq \beta < \alpha e^{-1/2\alpha}$. Then, $G'(m,h,\beta n,(1-\beta)n)$ contains a complex component with probability $O(1/n)$ for every fixed integer $h \in [e^{-1/\alpha}m-m^{2/3},e^{-1/\alpha}m+m^{2/3}]$.
\end{lemma}

\begin{proof}
By the definition of $G'(m,h,\beta n,(1-\beta)n)$, each vertex in $U'$ has degree at most one, thus removing $U'$ does not affect the existence of a complex component. Moreover, the graph $G'(m,h,\beta n,(1-\beta)n)$ with part $U'$ removed has exactly the same distribution as $G(m,h,\beta n)$ given in Definition \ref{defG1}. Therefore, it is sufficient to consider the graph $G(m,h,\beta n)$ instead.

Using the same technique as in Mahdian's proof of \cite[Lemma 4]{mahdian}, define a \textit{minimal bad graph} to be two vertices joined by three vertex-disjoint paths, or two vertex-disjoint cycles joined by a path which is also vertex-disjoint from the two cycles except at both endpoints (the path can be degenerate, which is the only exception that the two cycles share a vertex). Note that any proper subgraph of a minimal bad graph does not contain a complex component, and every graph that contains a complex component must contain a minimal bad graph as a subgraph.

Let $X$ and $Y$ be subsets of vertices of $G(m,h,\beta n)$ in $V$ and $U$, respectively. Define $BAD_{X,Y}$ to be an event that $X \cup Y$ contains a minimal bad graph as a \textit{spanning} subgraph. Then, let $p_1 = |X|$, $p_2 = |Y|$, and $p = p_1+p_2$. Observe that $BAD_{X,Y}$ can occur only when $|p_1 - p_2| \leq 1$, so $p_1,p_2 \geq \frac{p-1}{2}$. Also, there are at most $2p^2$ non-isomorphic minimal bad graphs with $p_1$ vertices in $V$ and $p_2$ vertices in $U$, with each of them having $p_1!p_2!$ ways to arrange the vertices, and there are at most $(p+1)!\binom{\beta n}{p+1}\left(\frac{1}{mh}\right)^{p+1}$ probability that all $p+1$ edges of each graph are selected in our random procedure. By the union bound, the probability of $BAD_{X,Y}$ is at most
$$2p^2p_1!p_2!(p+1)!\binom{\beta n}{p+1}\left(\frac{1}{mh}\right)^{p+1}
\leq 2p^2p_1!p_2!\left(\frac{\beta n}{mh}\right)^{p+1}.$$

Again, by the union bound, the probability that at least one $BAD_{X,Y}$ occurs is at most
\begin{align*}
\Pr\left[\bigvee_{X,Y} BAD_{X,Y}\right] &\leq \sum_{p_1,p_2} \binom{m}{p_1} \binom{h}{p_2} 2p^2p_1!p_2!\left(\frac{\beta n}{mh}\right)^{p+1} \\
&\leq \sum_{p_1,p_2} \frac{m^{p_1}}{p_1!} \cdot \frac{h^{p_2}}{p_2!} \cdot 2p^2p_1!p_2!\left(\frac{\beta}{\alpha h}\right)^{p+1} \\
&= \sum_{p_1,p_2} \frac{2p^2}{h} \left(\frac{\beta}{\alpha}\right)^{p+1} \left(\frac{m}{h}\right)^{p_1} \\
&\leq \sum_{p=1}^{\infty} \frac{O(p^2)}{n} \left(\frac{\beta}{\alpha}\right)^p \left(e^{-1/\alpha}-m^{-1/3}\right)^{-p/2} \\
&= \frac{O(1)}{n} \sum_{p=1}^{\infty} p^2\left(\frac{\alpha^2}{\beta^2}\left(e^{-1/\alpha}-m^{-1/3}\right)\right)^{-p/2}.
\end{align*}

By the assumption, we have $\alpha^2 e^{-1/\alpha} > \beta^2$, so $\frac{\alpha^2}{\beta^2}(e^{-1/\alpha}-m^{-1/3}) > 1$ for sufficiently large $m$, hence the above sum converges. Therefore, the probability that at least one $BAD_{X,Y}$ happens is at most $O(1/n)$.
\end{proof}

We can now prove the following theorem, which serves as an upper bound of $\alpha_k$.

\begin{theorem} \label{thmupperbound}
In a random instance with strict and $k$-incomplete preference lists, if $\alpha e^{-1/2\alpha} > 1-(1-e^{-1/\alpha})^{k-1}$, then a popular matching exists with probability $1-o(1)$.
\end{theorem}

\begin{proof}
Since $\alpha e^{-1/2\alpha} > 1-(1-e^{-1/\alpha})^{k-1}$, we can select a small enough $\delta_1 > 0$ such that $\alpha e^{-1/2\alpha} > 1-(1-e^{-1/\alpha})^{k-1}+\delta_1$. Let $J_1 = [(1-(1-e^{-1/\alpha})^{k-1}-\delta_1)n, (1-(1-e^{-1/\alpha})^{k-1}+\delta_1)n]$. From Lemma \ref{beta}, $|A_2| \in J_1$ with probability $1-o(1)$. Moreover, we have $\beta = \frac{t}{n} < \alpha e^{-1/2\alpha}$ for any integer $t \in J_1$.

Define $E_1$ to be an event that a popular matching exists in a random instance. First, consider the probability of $E_1$ conditioned on $|A_2| = t$ for each fixed integer $t \in J_1$. By Lemmas \ref{lemG2} and \ref{lemupperbound}, the top-choice graph contains a complex component with probability $O(n^{-1/3}) = o(1)$. Therefore, from Lemmas \ref{pmexists} and \ref{complexcomp2} we can conclude that a popular matching exists with probability $1-o(1)$, i.e. $\Pr\left[E_1 \big| |A_2| = t\right] = 1-o(1)$ for every fixed integer $t \in J_1$. So 
\begin{align*}
\Pr[E_1] &= \sum_t \Pr[|A_2| = t] \cdot \Pr\left[E_1\big||A_2| = t\right] \\
&\geq \sum_{t \in J_1} \Pr[|A_2| = t] \cdot \Pr\left[E_1\big||A_2| = t\right] \\
&\geq \Pr[|A_2| \in J_1] \cdot (1-o(1)) \\
&= (1-o(1))(1-o(1)) \\
&= 1-o(1).
\end{align*}

Hence, a popular matching exists with probability $1-o(1)$.
\end{proof}

\subsection{Lower Bound}
\begin{lemma} \label{lemlowerbound}
Suppose that $\alpha = m/n$ and $\alpha e^{-1/2\alpha} < \beta \leq 1$. Then, $G'(m,h,\beta n,(1-\beta)n)$ does not contain a complex component with probability $O(1/n)$ for every fixed integer $h \in [e^{-1/\alpha}m-m^{2/3},e^{-1/\alpha}m+m^{2/3}]$.
\end{lemma}

\begin{proof}
Again, by the same reasoning as in the proof of Lemma \ref{lemupperbound}, we can consider the graph $G(m,h,\beta n)$ instead of $G'(m,h,\beta n,(1-\beta)n)$, but now we are interested in an event that $G(m,h,\beta n)$ does not contain a complex component.

Since $\alpha e^{-1/2\alpha} < \beta$, for sufficiently small $\epsilon > 0$, we still have $\alpha e^{-1/2\alpha} < (1-\epsilon)^{3/2}\beta$. Consider the random bipartite graph $G(m,h,(1-\epsilon)\beta n)$ with parts $V$ having $m$ vertices and $U$ having $h$ vertices. For each vertex $v$, let a random variable $r_v$ be the degree of $v$. Since there are $(1-\epsilon)\beta n$ edges in the graph, the expected value of $r_v$ for each $v \in V$ is $c_1 = \frac{(1-\epsilon)\beta n}{m} = \frac{(1-\epsilon)\beta}{\alpha}$. Since $e^{-1/\alpha}m+m^{2/3} < \frac{e^{-1/\alpha}m}{1-\epsilon}$ for sufficiently large $m$, the expected value of $r_v$ for each $v \in U$ is
$$c_2 = \frac{(1-\epsilon)\beta n}{h} > \frac{(1-\epsilon)\beta n}{e^{-1/\alpha}m+m^{2/3}} > \frac{(1-\epsilon)\beta n}{e^{-1/\alpha}m/(1-\epsilon)} = \frac{(1-\epsilon)^2\beta}{\alpha e^{-1/\alpha}}$$
 for sufficiently large $m$. Furthermore, each $r_v$ has a binomial distribution, which converges to Poisson distribution when $m$ increases to infinity. The graph can be viewed as a special case of an \textit{inhomogeneous random graph} \cite{bollobas,soderberg}. With the assumption that $c_1c_2 > \frac{(1-\epsilon)^3\beta^2}{\alpha^2 e^{-1/\alpha}} > 1$, we can conclude that the graph contains a \textit{giant component} (a component containing a constant fraction of vertices of the entire graph) with probability $1-O(1/n)$, where the explanation is given in Appendix \ref{Appendix-4}.

Finally, consider the construction of $G(m,h,\beta n)$ by putting $\epsilon \beta n$ more random edges into $G(m,h,(1-\epsilon)\beta n)$. If two of those edges land in the giant component $C$, a complex component will be created. Since $C$ has size of a constant fraction of $m$, each edge has a constant probability to land in $C$, so the probability that at most one edge will land in $C$ is exponentially low. Therefore, $G(m,h,\beta n)$ does not contain a complex component with probability at most $O(1/n)$.
\end{proof}

We can now prove the following theorem, which serves as a lower bound of $\alpha_k$.

\begin{theorem} \label{thmlowerbound}
In a random instance with strict and $k$-incomplete preference lists, if $\alpha e^{-1/2\alpha} < 1-(1-e^{-1/\alpha})^{k-1}$, then a popular matching exists with \mbox{probability $o(1)$}.
\end{theorem}

\begin{proof}
Like in the proof of Theorem \ref{thmupperbound}, we can select a small enough $\delta_2 > 0$ such that $\alpha e^{-1/2\alpha} < 1-(1-e^{-1/\alpha})^{k-1}-\delta_2$. Let $J_2 = [(1-(1-e^{-1/\alpha})^{k-1}-\delta_2)n, (1-(1-e^{-1/\alpha})^{k-1}+\delta_2)n]$. We have $\frac{|A_2|}{n} \in J_2$ with probability $1-o(1)$ and $\beta = \frac{t}{n} > \alpha e^{-1/2\alpha}$ for any integer $t \in J_2$.

Now we define $E_2$ to be an event that a popular matching does not exist in a random instance. By the same reasoning as in the proof of Theorem \ref{thmupperbound}, we can prove that $\Pr\left[E_2\big||A_2| = t\right] = 1-o(1)$ for every fixed $t \in J_2$ and reach an analogous conclusion that $\Pr[E_2] = 1-o(1)$.
\end{proof}

\subsection{Phase Transition}

Since $f(x) = x e^{-1/2x} - (1-(1-e^{-1/x})^{k-1})$ is a strictly increasing function in $[1, \infty)$ for every $k \geq 1$, $f(x) = 0$ can have at most one root in $[1, \infty)$. That root, if exists, will serve as a phase transition point $\alpha_k$. In fact, for $k \geq 4$, $f(x) = 0$ has a unique solution in $[1, \infty)$; for $k \leq 3$, $f(x) = 0$ has no solution in $[1, \infty)$ and $\alpha e^{-1/2\alpha} > 1-(1-e^{-1/\alpha})^{k-1}$ for every $\alpha \geq 1$, so a popular matching always exists with high probability without a phase transition regardless of value of $\alpha$. Therefore, from Theorems \ref{thmupperbound} and \ref{thmlowerbound} we can conclude our main theorem below.

\begin{theorem} \label{mainthm}
In a random instance with strict and $k$-incomplete preference lists with $k \geq 4$, if $\alpha > \alpha_k$, where $\alpha_k \geq 1$ is the root of equation $x e^{-1/2x} = 1-(1-e^{-1/x})^{k-1}$, then a popular matching exists with probability $1-o(1)$; and if $\alpha < \alpha_k$, then a popular matching exists with probability $o(1)$. In such a random instance with $k \leq 3$, a popular matching exists with probability $1-o(1)$ for any value of $\alpha \geq 1$.
\end{theorem}

\section{Conclusion and Future Work}

For each value of $k \geq 4$, the phase transition occurs at the root $\alpha_k \geq 1$ of equation $x e^{-1/2x} = 1-(1-e^{-1/x})^{k-1}$ as shown in Figure \ref{figure-1}. Note that as $k$ increases, the right-hand side of the equation converges to 1, hence $\alpha_k$ converges to Mahdian's value of $\alpha_* \approx 1.42$ in the case with complete preference lists.

\begin{figure}[h]
\captionsetup{width=0.9\textwidth}
\centering
\includegraphics[width=110mm]{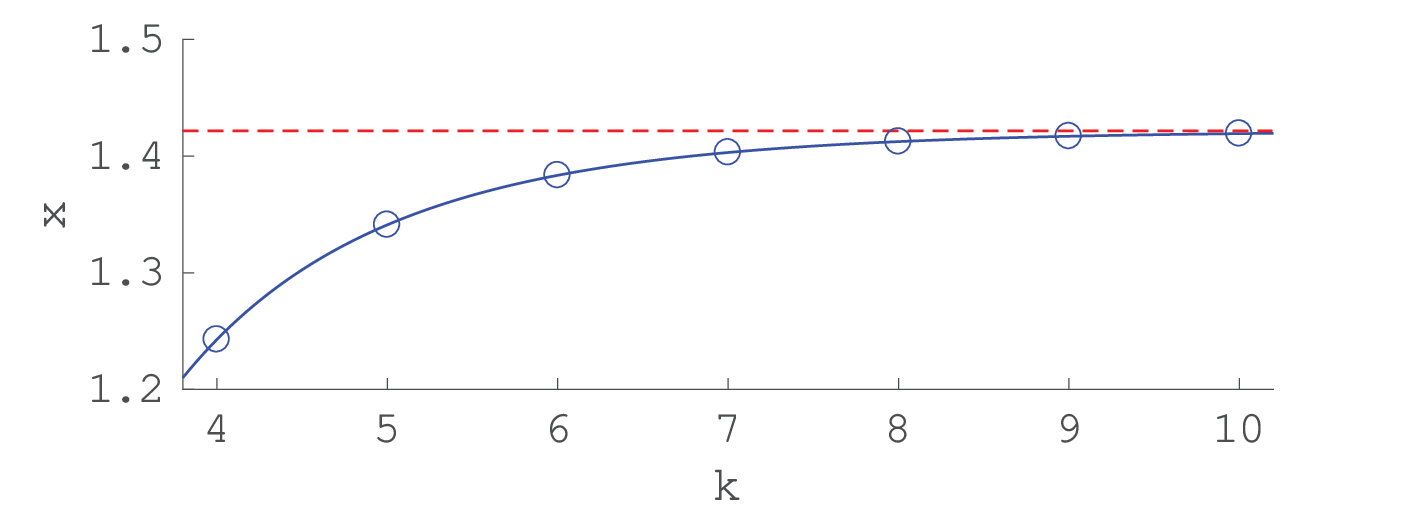}
\caption{Solution in $[1, \infty)$ of the equation $x e^{-1/2x} = 1-(1-e^{-1/x})^{k-1}$ for each $k \geq 4$, with the dashed line plotting $x = \alpha_* \approx 1.42$} \label{figure-1}
\end{figure}

\begin{remark}
For each person $a$, as the length of $P_a$ increases, the probability that $P_a \nsubseteq F$ and thus $a \in A_2$ also increases, and so do the expected size of $A_2$ and the phase transition point. Therefore, in the case that the lengths of people's preference lists are fixed but not equal (e.g. half of the people have preference lists with length $k_1$, and another half have those with length $k_2$), the phase transition will occur between $\alpha_{k_{\text{min}}}$ and $\alpha_{k_{\text{max}}}$, where $k_{\text{min}}$ and $k_{\text{max}}$ are the shortest and longest lengths of people's preference lists, respectively.
\end{remark}

In many real-world situations, ties can and are likely to occur among people's preference lists. \textsc{rpmp} in the case with ties allowed was mentioned by Mahdian \cite{mahdian} and simulated by Abraham et al. \cite{abraham} using a parameter $t$ to denote the probability that each entry in a preference list is tied with previous entry. Intuitively, and also confirmed by the experimental results of \cite{abraham}, when ties are very likely to occur ($t$ is very close to 1), a popular matching is likely to exist even when $\alpha = 1$. However, the transition point for each value of $t$ has still not been found yet. A possible future work is to study the transition point in this case for each value of $t$, both with complete and incomplete preference lists. Another interesting generalization of \textsc{rpmp} is the Capacitated House Allocation Problem, where each item can be matched with more than one person. A possible future work is to find the transition point in the most basic case where every item has the same capacity $c$.

\newpage
\appendix
\section{Proof of Inequalities (\ref{ineq1}) and (\ref{ineq2})} \label{Appendix-1}

For $k \geq 2$, we will prove that $c_1 = \frac{c_2}{(k-1)(c_2+2)}$ satisfies inequalities (\ref{ineq1}) and (\ref{ineq2}).

Let $p = 1-e^{-1/\alpha}$. We have $0 < p < 1$ and $0 < c_1 < 1$. So,
\begin{align*}
(p-c_1)^{k-1} &= p^{k-1} - \binom{k-1}{1} p^{k-2}c_1 + \binom{k-1}{2} p^{k-3}c_1^2 - \cdots + (-1)^{k-1}\binom{k-1}{k-1} c_1^{k-1} \\
&\geq p^{k-1} - \left[(k-1) c_1 + (k-1)^2 c_1^2 + \cdots + (k-1)^{k-1} c_1^{k-1}\right] \\
&= p^{k-1} - \left[\frac{c_2}{c_2+2} + \left(\frac{c_2}{c_2+2}\right)^2 + \cdots + \left(\frac{c_2}{c_2+2}\right)^{k-1}\right] \\
&> p^{k-1} - \left[\frac{c_2}{c_2+2} + \left(\frac{c_2}{c_2+2}\right)^2 + \cdots\right] \\
&= p^{k-1} - \frac{\frac{c_2}{c_2+2}}{1-\frac{c_2}{c_2+2}} \\
&= p^{k-1} - \frac{c_2}{2}.
\end{align*}

Therefore $(1-e^{-1/\alpha}-c_1)^{k-1} > (1-e^{-1/\alpha})^{k-1} - \frac{c_2}{2}$. Also, we have
\begin{align*}
(p+c_1)^{k-1} &= p^{k-1} + \binom{k-1}{1} p^{k-2}c_1 + \binom{k-1}{2} p^{k-3}c_1^2 + \cdots + \binom{k-1}{k-1} c_1^{k-1} \\
&\leq p^{k-1} + (k-1) c_1 + (k-1)^2 c_1^2 + \cdots + (k-1)^{k-1} c_1^{k-1} \\
&= p^{k-1} + \frac{c_2}{c_2+2} + \left(\frac{c_2}{c_2+2}\right)^2 + \cdots + \left(\frac{c_2}{c_2+2}\right)^{k-1} \\
&< p^{k-1} + \frac{c_2}{c_2+2} + \left(\frac{c_2}{c_2+2}\right)^2 + \cdots \\
&= p^{k-1} + \frac{\frac{c_2}{c_2+2}}{1-\frac{c_2}{c_2+2}} \\
&= p^{k-1} + \frac{c_2}{2}.
\end{align*}

Therefore $(1-e^{-1/\alpha}+c_1)^{k-1} > (1-e^{-1/\alpha})^{k-1} + \frac{c_2}{2}$.

\section{Explanation of the Lower Bound} \label{Appendix-4}

An inhomogeneous random graph is a generalization of an Erd\H{o}s-R\'{e}nyi graph, where vertices of the graph are divided into several (finite or infinite) types. Each vertex of type $i$ has $\kappa_{ij}$ expected neighbors of type $j$.

The bipartite graph $G(m,h,(1-\epsilon)\beta n)$ can be considered as a special case of the inhomogeneous random graph where there are two types of vertices, with $\kappa_{11} = 0$, $\kappa_{12} = c_1$, $\kappa_{21} = c_2$, and $\kappa_{22} = 0$. It has an \textit{offspring matrix}
$$T_\kappa = \{\kappa_{ij}\}_{i,j=1}^2 = \begin{bmatrix} 0 & c_1 \\ c_2 & 0 \end{bmatrix},$$
which has the largest eigenvalue $||T_\kappa|| = \sqrt{c_1c_2} > 1$. This is a necessary and sufficient condition to conclude that $G(m,h,(1-\epsilon)\beta n)$ contains a giant component with $1-o(1)$ probability \cite{bollobas,soderberg}. In fact, by giving a precise bound in each step of \cite{bollobas}, it is possible to show that the probability is greater than $1-O(1/n)$ as desired.

Alternatively, we hereby show a direct proof of the bipartite case by approximating the construction of the graph with a Galton-Watson branching process similar to that in the proof of existence of a giant component in the Erd\H{o}s-R\'{e}nyi graph in \cite[pp.182-192]{alon}.

The Galton-Watson branching process is a process that generates a random graph in a breadth-first search tree manner when given a starting vertex and a distribution of the degree of each vertex. The process begins when the starting vertex spawns a number of children which are put in the queue in some order. Then, the first vertex in the queue also spawns children which are put at the end of the queue by the same manner, and so on. The process may stop at some point when the queue becomes empty, or otherwise continues indefinitely.

Consider the construction of $G(m,h,(1-\epsilon)\beta n)$ with parts $V$ and $U$ starting at a vertex and discovering new vertices in a breadth-first search tree manner. We approximate it with the Galton-Watson branching process. Let $T$ be the size of the process ($T = \infty$ if the process continues forever). Let $z_1$ and $z_2$ be the probability that $T < \infty$ when starting the process at a vertex in $V$ and $U$, respectively. Also, let $Z_1$ and $Z_2$ be the number of children the root has when starting the process at a vertex in $V$ and $U$, respectively.

Given that the root has $i$ children, in order for the branching process to be finite, all of the $i$ branches must be finite, so we get the equations.
\begin{align*}
z_1 &= \sum_{i=0}^{\infty} \Pr[Z_1 = i]z_2^i; \\
z_2 &= \sum_{i=0}^{\infty} \Pr[Z_2 = i]z_1^i.
\end{align*}

Therefore,
\begin{align*}
z_1 &= \sum_{i=0}^{\infty} \frac{c_1^i e^{-c_1}}{i!} \left(\sum_{j=0}^{\infty} \frac{c_2^j e^{-c_2} z_1^j}{j!}\right)^i \\
&= \sum_{i=0}^{\infty} \frac{c_1^i e^{-c_1}}{i!} e^{c_2(z_1-1)i} \\
&= e^{c_1(e^{c_2(z_1-1)}-1)}.
\end{align*}

Setting $y = 1-z_1$ yields the equation
\begin{equation} \label{eq1}
1-y = e^{c_1(e^{-c_2y}-1)}.
\end{equation}

Define $g(y) = 1-y-e^{c_1(e^{-c_2y}-1)}$. We have $g(0) = 1-0-1 = 0$, $g(1) < 0$, and $g'(0) = c_1c_2-1$. By the assumption that $c_1c_2 > 1$, we have $g'(0) > 0$, so there must be $y \in (0,1)$ such that $g(y) = 0$, thus being a solution of (\ref{eq1}). So, $\Pr[T = \infty] = y \in (0,1)$, when $y$ is a solution of (\ref{eq1}), meaning that there is a constant probability that the process continues indefinitely. Moreover, from the property of Poisson distribution we can show that $\Pr[x < T < \infty]$ is exponentially low in term of $x$. Therefore, we can select a constant $k_1$ such that $\Pr[k_1 \log n < T < \infty] < O(1/n^2)$.

Finally, when we perform the Galton-Watson branching process at a vertex in $G(m,h,(1-\epsilon)\beta n)$, there is a constant probability that the process will continue indefinitely, thus creating a giant component. Otherwise, with probability $1-O(1/n^2)$ we will create a component with size smaller than $k_1 \log n$, so we can remove that component from the graph and then repeatedly perform the process starting at a new vertex. After repeatedly performing this process for some logarithmic number of times, we only remove $O(\log^2 n)$ vertices from the graph, which does not affect the constant $y = \Pr[T = \infty]$, so the probability that we never end up with a giant component in every time is at most $O(1/n)$. Therefore, $G(m,h,(1-\epsilon)\beta n)$ contains a giant component with probability $1-O(1/n)$.

\begin{remark}
In the complete preference lists setting with $\alpha e^{-1/2\alpha} < (1-\epsilon)^{3/2}$, we have $c_1 = \frac{1-\epsilon}{\alpha}$ and $c_2 > \frac{(1-\epsilon)^2}{\alpha e^{-1/\alpha}}$, which we still get $c_1c_2 = \frac{(1-\epsilon)^3}{\alpha^2 e^{-1/\alpha}} > 1$, which is a sufficient condition to reach the same conclusion.
\end{remark}

\end{document}